\documentclass[a4paper,12pt]{elsarticle}
\usepackage[T1]{fontenc}
\usepackage[english]{babel}
\usepackage{ae,aecompl}
\usepackage{amsmath, amsthm}
\usepackage{amssymb}
\usepackage[utf8]{inputenc}
\usepackage[section] {placeins}
\usepackage[ruled, vlined, linesnumbered]{algorithm2e}
\newtheorem {Theorem}                 {Theorem}         [section]
\newtheorem {theorem}      [Theorem]  {Theorem}
\newtheorem {myalgorithm}    [Theorem]  {Algorithm}

\newtheorem {lemma}        [Theorem]  {Lemma}

\input amssym.def
\input amssym

\usepackage{pgfplots}
\usepackage{verbatim}
\usepackage{subfigure}
\usepackage{tikz}
\usetikzlibrary{shapes,arrows,backgrounds,mindmap}
\usepackage{titlesec}
\journal{arXiv}

\begin{document}
	\begin{frontmatter}
		\title{Minimum $2$-vertex strongly biconnected spanning directed subgraph problem}
		\author{Raed Jaberi}
		
		\begin{abstract}  
	    A directed graph $G=(V,E)$ is strongly biconnected if $G$ is strongly connected and its underlying graph is biconnected. A strongly biconnected directed graph $G=(V,E)$ is called $2$-vertex-strongly biconnected if $|V|\geq 3$ and the induced subgraph on $V\setminus\left\lbrace w\right\rbrace $ is strongly biconnected for every vertex $w\in V$. In this paper we study the following problem.		
		Given a $2$-vertex-strongly biconnected directed graph $G=(V,E)$, compute an edge subset $E^{2sb} \subseteq E$ of minimum size such that the subgraph $(V,E^{2sb})$ is $2$-vertex-strongly biconnected.
		\end{abstract} 
	\end{frontmatter}
	\section{Introduction}
	  A directed graph $G=(V,E)$ is strongly biconnected if $G$ is strongly connected and its underlying graph is biconnected. A strongly biconnected directed graph $G=(V,E)$ is called $k$-vertex-strongly biconnected if $|V|> k$ and for each $L\subset V$ with $|L|<k$, the induced subgraph on $V\setminus L$ is strongly biconnected. The minimum $k$-vertex-strongly biconnected spanning subgraph problem (denoted by MKVSBSS) is formulated as follows.		
	Given a $k$-vertex-strongly biconnected directed graph $G=(V,E)$, compute an edge subset $E^{ksb} \subseteq E$ of minimum size such that the subgraph $(V,E^{ksb})$ is $k$-vertex-strongly biconnected. In this paper we consider the MKVSBSS problem for $k=2$.
 Each $2$-vertex-strongly-biconnected directed graph is $2$-vertex-connected, but the converse is not necessarily true. Thus, optimal solutions for minimum $2$-vertex-connected spanning subgraph (M2VCSS) problem are not necessarily feasible solutions for the $2$-vertex strongly biconnnected spanning subgraph problem, as shown in Figure \ref{figure:sbiconnectedexampleoptimalsolutions}.

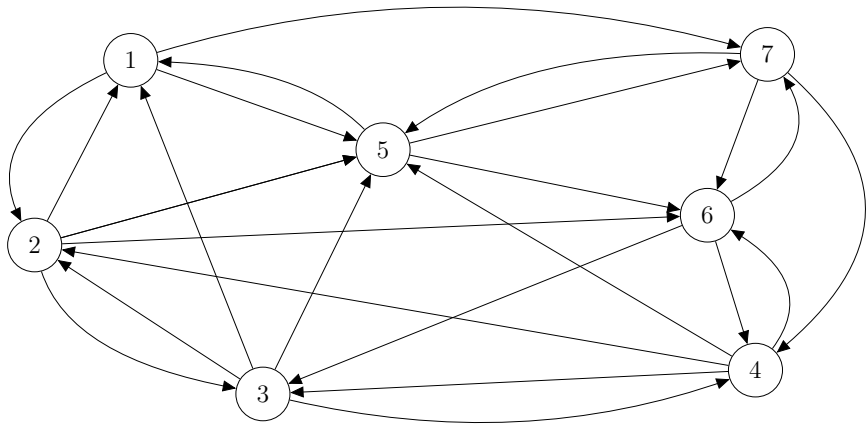
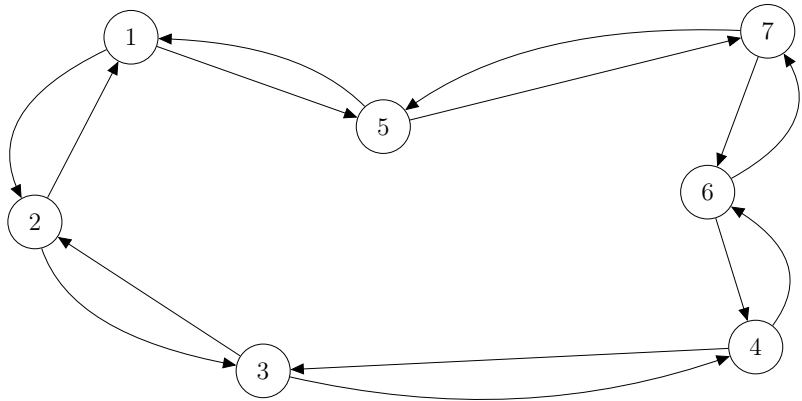
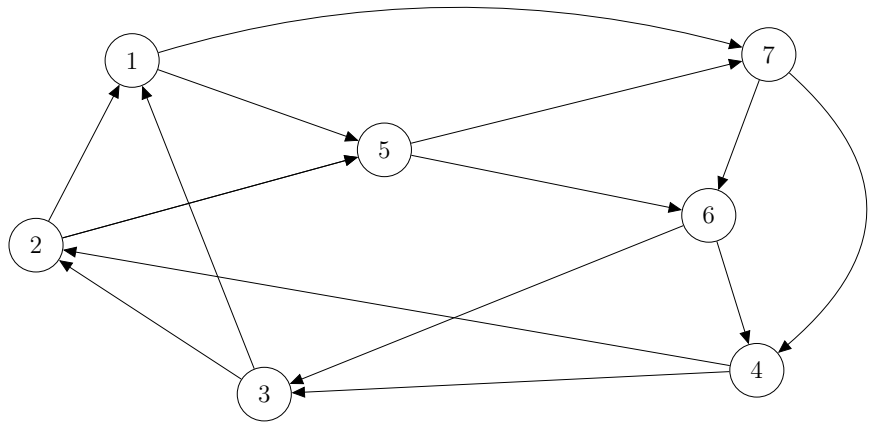
\begin{figure}[htp]
	\centering
	
	\subfigure[]{
	
\scalebox{0.79}{

		\begin{tikzpicture}[xscale=2]
		\tikzstyle{every node}=[color=black,draw,circle,minimum size=0.9cm]
		\node (v1) at (-1.6,3.1) {$1$};
		\node (v2) at (-2.4,0) {$2$};
		\node (v3) at (-0.5, -2.5) {$3$};
		\node (v4) at (3.6,-2.1) {$4$};
		\node (v5) at (0.5,1.6) {$5$};
		\node (v6) at (3.2,0.5) {$6$};
		\node (v7) at (3.7,3.2) {$7$};
	
		\begin{scope}   
		\tikzstyle{every node}=[auto=right]   
		\draw [-triangle 45] (v2) to (v5);
			\draw [-triangle 45] (v2) to (v6);
				\draw [-triangle 45] (v3) to (v5);
	\draw [-triangle 45] (v2) to (v1);
	\draw [-triangle 45] (v1) to[bend right] (v2);
	\draw [-triangle 45] (v1) to (v5);
	\draw [-triangle 45] (v5) to[bend right] (v1);
	\draw [-triangle 45] (v5) to (v7);
	\draw [-triangle 45] (v7) to[bend right] (v5);
	\draw [-triangle 45] (v7) to (v6);
	\draw [-triangle 45] (v6) to[bend right] (v7);
	\draw [-triangle 45] (v6) to (v4);
	\draw [-triangle 45] (v4) to[bend right] (v6);
	\draw [-triangle 45] (v4) to (v3);
	\draw [-triangle 45] (v3) to[bend right] (v4);
	\draw [-triangle 45] (v3) to (v2);
	\draw [-triangle 45] (v2) to[bend right] (v3);
	\draw [-triangle 45] (v1) to[bend left] (v7);
	\draw [-triangle 45] (v5) to (v6);
	\draw [-triangle 45] (v6) to (v3);
	\draw [-triangle 45] (v4) to (v2);
	\draw [-triangle 45] (v3) to (v1);
	\draw [-triangle 45] (v4) to (v5);
	\draw [-triangle 45] (v2) to (v5);
	\draw [-triangle 45] (v7) to [bend left](v4);
		\end{scope}
		\end{tikzpicture}
	}
	}
	\subfigure[]{
		\scalebox{0.79}{
		\begin{tikzpicture}[xscale=2]
		\tikzstyle{every node}=[color=black,draw,circle,minimum size=0.9cm]
	\node (v1) at (-1.6,3.1) {$1$};
\node (v2) at (-2.4,0) {$2$};
\node (v3) at (-0.5, -2.5) {$3$};
\node (v4) at (3.6,-2.1) {$4$};
\node (v5) at (0.5,1.6) {$5$};
\node (v6) at (3.2,0.5) {$6$};
\node (v7) at (3.7,3.2) {$7$};
		
		\begin{scope}   
		\tikzstyle{every node}=[auto=right]   
		
		\draw [-triangle 45] (v2) to (v1);
		\draw [-triangle 45] (v1) to[bend right] (v2);
		\draw [-triangle 45] (v1) to (v5);
		\draw [-triangle 45] (v5) to[bend right] (v1);
		\draw [-triangle 45] (v5) to (v7);
		\draw [-triangle 45] (v7) to[bend right] (v5);
		\draw [-triangle 45] (v7) to (v6);
		\draw [-triangle 45] (v6) to[bend right] (v7);
		\draw [-triangle 45] (v6) to (v4);
		\draw [-triangle 45] (v4) to[bend right] (v6);
		\draw [-triangle 45] (v4) to (v3);
		\draw [-triangle 45] (v3) to[bend right] (v4);
		\draw [-triangle 45] (v3) to (v2);
		\draw [-triangle 45] (v2) to[bend right] (v3);
	
		\end{scope}
		\end{tikzpicture}}}
\subfigure[]{
	\scalebox{0.79}{
	\begin{tikzpicture}[xscale=2]
	\tikzstyle{every node}=[color=black,draw,circle,minimum size=0.9cm]
	\node (v1) at (-1.6,3.1) {$1$};
\node (v2) at (-2.4,0) {$2$};
\node (v3) at (-0.5, -2.5) {$3$};
\node (v4) at (3.6,-2.1) {$4$};
\node (v5) at (0.5,1.6) {$5$};
\node (v6) at (3.2,0.5) {$6$};
\node (v7) at (3.7,3.2) {$7$};
	
	\begin{scope}   
	\tikzstyle{every node}=[auto=right]   
	\draw [-triangle 45] (v2) to (v5);
	\draw [-triangle 45] (v2) to (v1);

	\draw [-triangle 45] (v1) to (v5);
	
	\draw [-triangle 45] (v5) to (v7);
	
	\draw [-triangle 45] (v7) to (v6);

	\draw [-triangle 45] (v6) to (v4);

	\draw [-triangle 45] (v4) to (v3);

	\draw [-triangle 45] (v3) to (v2);
	
	\draw [-triangle 45] (v1) to[bend left] (v7);
	\draw [-triangle 45] (v5) to (v6);
	\draw [-triangle 45] (v6) to (v3);
	\draw [-triangle 45] (v4) to (v2);
	\draw [-triangle 45] (v3) to (v1);
%	\draw [-triangle 45] (v4) to (v5);
	\draw [-triangle 45] (v2) to (v5);
	\draw [-triangle 45] (v7) to [bend left](v4);
	\end{scope}
		\end{tikzpicture}}}
\caption{(a) A $2$-vertex strongly biconnected graph. (b) An optimal solution for the minimum $2$-vertex-connected spanning subgraph problem. But note that this subgraph is not $2$-vertex strongly biconnected.(c) An optimal solution for the minimum $2$-vertex strongly biconnected spanning subgraph problem}
\label{figure:sbiconnectedexampleoptimalsolutions}
\end{figure}

The problem of finding a $k$-vertex-connected spanning subgraph of a $k$-vertex-connected directed graph is NP-hard for $k\geq 1$ \cite{G79}. Results of Edmonds \cite{Edmonds72} and Mader \cite{Mader85} imply that the number of edges in each minimal $k$-vertex-connected directed graph is at most $2kn$ \cite{CT00}. Cheriyan and Thurimella \cite{CT00} gave a $(1+1/k)$-approximation algorithm for the minimum $k$-vertex-connected spanning subgraph problem. Georgiadis \cite{Georgiadis11} improved the running time of this algorithm for the M2VCSS problem and presented a linear time approximation algorithm that achieves an approximation factor of $3$ for the M2VCSS problem. Georgiadis et al. \cite{GIK20} provided linear time $3$-and $2$- approximation algorithms based on the results of \cite{GT16,G10,ILS12,FILOS12} for the M2VCSS problem. Furthermore, Georgiadis et al. \cite{GIK20} improved the algorithm of Cheriyan and Thurimella when $k=2$. Strongly connected components of a directed graph and blocks of an undirected graphs can be found in linear time using Tarjan's algorithm \cite{TAARJAN72}.
 Wu and Grumbach \cite{WG2010} introduced the concept of strongly biconnected directed graph and strongly connected components.  The MKVSBSS problem is NP-hard for $k\geq 1$. In this paper we study the MKVSBSS problem when $k=2$ (denoted by M2VSBSS). 
	
\section{Approximation algorithm for the M2VSBSS problem} 
In this section we present an approximation algorithm (Algorithm \ref{algo:approximationalgorithmfor2sb}) for the M2VSBSS Problem. This algorithm is based on b-articulation points, minimal 2-vertex-connected subgraphs, and Lemma \ref{def:addingedgestoreducesbcs}. A vertex $w$ in a strongly biconnected directed graph $G$ is a b-articulation points if $G\backslash\left\lbrace w \right\rbrace $ is not strongly biconnected \cite{Jaberi20}.

\begin{lemma} \label{def:addingedgestoreducesbcs}
	Let $G_s=(V,E_s)$ be a subgraph of a strongly biconnected directed graph $G=(V,E)$ such that $G_s$ is strongly connected and $G_s$ has $t>0$ strongly biconnected components. Let $(u,w)$ be an edge in $E\setminus E_s$  such that $u,w $ are in not in the same strongly
	biconnected component of $G_s$. Then the directed subgraph $(V,E\cup \left\lbrace (u,w) \right\rbrace )$ contains at most $t-1$ strongly biconnected components.
\end{lemma}
\begin{proof}
	Since $G_s$ is strongly connected, there exists a simple path $p$ from $w$ to $u$ in $G_s$. Path $p$ and edge $(u,w)$ form a simple cycle. Consequently, the vertices $u,w$ are in the same strongly biconnected component of the subgraph $(V,E\cup \left\lbrace (u,w) \right\rbrace )$.
\end{proof}

\begin{figure}[htbp]
	\begin{myalgorithm}\label{algo:approximationalgorithmfor2sb}\rm\quad\\[-5ex]
		\begin{tabbing}
			\quad\quad\=\quad\=\quad\=\quad\=\quad\=\quad\=\quad\=\quad\=\quad\=\kill
			\textbf{Input:} A $2$-vertex strongly biconnected directed graph $G=(V,E)$ \\
			
			\textbf{Output:} a $2$-vertex strongly biconnected subgraph $G_{2s}=(V,E_{2s})$\\
			{\small 1}\> find a minimal $2$-vertex-connected subgraph $G_{1}=(V,E_{1})$ of $G$.\\
			{\small 2}\> \textbf{if} $G_{1}$ is $2$-vertex strongly biconnected \textbf{then} \\
			{\small 3}\>\>output  $G_{1}$\\
			{\small 4}\> \textbf{else}\\
			{\small 5}\>\>$E_{2s} \leftarrow E_{1}$\\
			{\small 6}\>\>  $G_{2s}\leftarrow (V,E_{2s})$ \\
			{\small 7}\>\>  identify the b-articulation points of $G_{1}$.\\
			{\small 8}\>\>  \textbf{for} evry b-articulation point $b\in V$ \textbf{do} \\
			{\small 9}\>\>\>  \textbf{while} $G_{2s}\setminus\left\lbrace b \right\rbrace $ is not  strongly biconnected \textbf{do}\\ 
		{\small 10}\>\>\>\>  calculate the  strongly biconnected components of $G_{2s}\setminus\left\lbrace b \right\rbrace $\\
		{\small 11}\>\>\>\> find an edge $(u,w) \in E\setminus E_{2s}$ such that $u,w $ are not in \\
		{\small 12}\>\>\>\>the same strongly biconnected components of $G_{2s}\setminus\left\lbrace b \right\rbrace $.\\
		{\small 13}\>\>\>\> $E_{2s} \leftarrow E_{2s} \cup\left\lbrace  (u,w)\right\rbrace $ \\
		{\small 14}\>\>\>output $G_{2s}$
	
		\end{tabbing}
	\end{myalgorithm}
\end{figure}

\begin{lemma} 
Algorithm \ref{algo:approximationalgorithmfor2sb} returns a $2$-vertex strongly biconnected directed subgraph.
\end{lemma}	
\begin{proof}
It follows from Lemma \ref{def:addingedgestoreducesbcs}.
\end{proof}

The following lemma shows that each optimal solution for the M22VSBSS problem has at least $2n$ edges.
\begin{lemma} \label{def:optsolutionatlesttwon}
	Let $G=(V,E)$ be a $2$-vertex-strongly biconnected directed graph. Let $O\subseteq E$ be an optimal solution for the M2VSBSS problem. Then $|O|\geq 2n$. 
	
\end{lemma}
\begin{proof}
	for any vertex $x\in V$, the removal of $x$ from the subgraph $(V,O)$ leaves a strongly biconnected directed subgraph. Since each strongly biconnected directed graph is stronly connected,
	the subgraph $(V,O)$ has no strong articulation points. Therefore, the directed subgraph $(V,O)$ is $2$-vertex-connected.
\end{proof}
Let $l$ be the number of b-articulation points in  $G_{1}$. The following lemma shows that Algorithm \ref{algo:approximationalgorithmfor2sb} has an approximation factor of $(2+l/2)$.
\begin{theorem}
	Let $l$ be the number of b-articulation points in  $G_{1}$. Then, $|E_{2s}|\leq l(n-1)+4n$.
\end{theorem}
\begin{proof}
	Results of Edmonds \cite{Edmonds72} and Mader \cite{Mader85} imply that $|E_1|\leq 4n$  \cite{CT00,Georgiadis11}. Moreover, by Lemma \ref{def:optsolutionatlesttwon}, every optimal solution for the M22VSBSS problem has size at least $2n$. For every b-articulation point in line $8$,  Algorithm \ref{algo:approximationalgorithmfor2sb} adds at most $n-1$ edge to $E_{2s}$ in while loop. Therefore, $|E_{2s}|\leq l(n-1)+4n$
\end{proof}
\begin{Theorem}
	The running time of Algorithm \ref{algo:approximationalgorithmfor2sb} is $O(n^{2}m)$.
\end{Theorem}
\begin{proof}
	A minimal $2$-vertex-connected subgraph can be found in time $O(n^2)$ \cite{Georgiadis11,GIK20}. B-articulation points can be computed in $O(nm)$ time.
	The strongly biconnected components of a directed graph can be identified in linear time \cite{WG2010}. 
	Furthermore, by Lemma \ref{def:addingedgestoreducesbcs}, lines $9$--$13$ take $O(nm)$ time.
\end{proof}

\section{Open Problems}
 Results of Mader \cite{Mader71,Mader72} imply that the number of edges in each minimal $k$-vertex-connected undirected graph is at most $kn$ \cite{CT00}. Results of Edmonds \cite{Edmonds72} and Mader \cite{Mader85} imply that the number of edges in each minimal $k$-vertex-connected directed graph is at most $2kn$ \cite{CT00}. These results imply a $2$-approximation algorithm \cite{CT00} for minimum $k$-vertex-connected spanning subgraph problem for undirected and directed graphs \cite{CT00} because every vertex in a $k$-vertex-connected undirected graphs has degree at least $k$ and every vertex in a $k$-vertex-connected directed graph has outdegree at least $k$ \cite{CT00}. Note that these results imply a $7/2$ approximation algorithm for the M2VSBSS problem by calculating a minimal $2$-vertex-connected directed subgraph of a $2$-vertex strongly biconnected directed graph $G=(V,E)$ and a minimal $3$-vertex connected undirected subgraph of the underlying graph of $G$.
 
 \begin{lemma}
Let $G=(V,E)$ be a $2$-vertex strongly biconnected directed graph. Let $G_1=(V,L)$ be a minimal $2$-vertex-connected subgraph of $G$ and let $G_2=(V,U)$ be a minimal $3$-vertex-connected subgraph of the underlying graph of $G$. Then the directed subgraph $G_s=(V,L\cup A )$ is $2$-vertex strongly connected, where $A=\left\lbrace (v,w) \in E \text{ and } (v,w) \in U \right\rbrace $. Moreover, $|L\cup A|\leq 7n$
 \end{lemma}
\begin{proof}
	Let $w$ be any vertex of the subgraph $G_s$. Since the $G_1=(V,L)$ is $2$-vertex-connected, subgraph $G_s$ has no strong articulation points. Therefore, $G_s\setminus\left\lbrace w\right\rbrace $ is strongly connected. Moreover, the underlying graph of  $G_s\setminus\left\lbrace w\right\rbrace $ is biconnected because the underlying graph of $G_s$ is $3$-vertex-connected. Results of Edmonds \cite{Edmonds72} and Mader \cite{Mader85} imply that $|L|<4n$. Results of Mader \cite{Mader71,Mader72} imply that $|U|\leq 3n$.
\end{proof}
 
 An open problem is whether each minimal $2$-vertex strongly biconnected directed graph has at most $4n$ edges.

 Cheriyan and Thurimella \cite{CT00} presented a $(1+1/k)$-approximation algorithm for the minimum $k$-vertex-connected spanning subgraph problem for directed and undirected graphs. The algorithm of Cheriyan and Thurimella \cite{CT00} has an approximation factor of $3/2$ for the minimum $2$-vertex-connected directed subgraph problem. Let $G=(V,E)$ be a $2$-vertex strongly biconnected directed graph and let $E^{CT}$ be the output of the algorithm of Cheriyan and Thurimella \cite{CT00}. The directed subgraph $(V,E^{CT})$ is not necessarily $2$-vertex strongly biconnected. But a $2$-vertex strongly biconnected subgraph can be obtained by performing the following third phase. For each edge $e\in E\setminus E^{CT}$, if the underlying graph of $G\setminus \left\lbrace e\right\rbrace $ is $3$-vertex-connected, delete $e$ from $G$. We leave as open problem whether this algorithm has an approximation factor of $3/2$ for the M2VSBSS problem.


\begin{thebibliography}{4}
	\bibitem{CT00} J. Cheriyan, R. Thurimella,
	Approximating Minimum-Size $k$-Connected Spanning Subgraphs via Matching. SIAM J. Comput. $30(2): 528$--$560 (2000)$
	\bibitem{Edmonds72} J. Edmonds, Edge-disjoint branchings. Combinatorial Algorithms, pages $91$--$96$,
	$1972$
	\bibitem {FILOS12} D. Firmani, G.F. Italiano, L. Laura, A. Orlandi, F. Santaroni, Computing strong articulation points and strong bridges in large scale graphs, SEA, LNCS $7276$, ($2012$) $195$--$207$.
	\bibitem{G79} M. R. Garey, David S. Johnson:
	Computers and Intractability: A Guide to the Theory of NP-Completeness. W. H. Freeman 1979, ISBN $0$--$7167$--$1044$--$7$
	\bibitem {G10} L. Georgiadis, Testing $2$-vertex connectivity and computing pairs of vertex-disjoint s-t paths in digraphs, In Proc. $37$th ICALP, Part I, LNCS $6198$ ($2010$) $738$--$749$.
	\bibitem{GT16}Loukas Georgiadis, Robert E. Tarjan, Dominator Tree Certification and Divergent Spanning Trees. ACM Trans. Algorithms $12(1): 11:1$--$11:42 (2016)$
	\bibitem{Georgiadis11} L. Georgiadis:
	Approximating the Smallest 2-Vertex Connected Spanning Subgraph of a Directed Graph. ESA $2011: 13$--$24$
	\bibitem{GIK20}L. Georgiadis, G. F. Italiano, A. Karanasiou:
	Approximating the smallest 2-vertex connected spanning subgraph of a directed graph. Theor. Comput. Sci. $807$: $185$--$200 (2020)$
	\bibitem {ILS12} G.F. Italiano, L. Laura, F. Santaroni,
	Finding strong bridges and strong articulation points in linear time, Theoretical Computer Science $447$ ($2012$) $74$--$84$.
	\bibitem{Jaberi20}R. Jaberi,
	b-articulation points and b-bridges in strongly biconnected directed graphs,  CoRR abs/2007.01897 $(2020)$
	\bibitem{Mader85} W. Mader,
	Minimal $n$-fach zusammenhängende Digraphen. J. Comb. Theory, Ser. B $38(2): 102$--$117 (1985)$
	\bibitem{Mader71} W. Mader, Minimale n-fach kantenzusammenhängende Graphen. Math. Ann., $191:21$
	--$28, 1971$
	\bibitem{Mader72} W. Mader, Ecken vom Grad n in minimalen n-fach zusammenhängenden Graphen. Arch. Math. (Basel), $23:219$--$224, 1972$
	\bibitem{TAARJAN72}
	R. E. Tarjan, Depth First Search and Linear Graph Algorithms, SIAM J. Comput.,$1(2) (1972),146$--$160$
	\bibitem{WG2010} Z. Wu, S. Grumbach,
Feasibility of motion planning on acyclic and strongly connected directed graphs. Discret. Appl. Math. $158(9): 1017$--$1028 (2010)$	
\end{thebibliography}
\end{document}